\PassOptionsToPackage{colorlinks=true,allcolors=black}{hyperref}

\documentclass[11pt]{article}

\usepackage[T1]{fontenc}
\usepackage{lmodern}
\usepackage{setspace}
\usepackage{amsmath,amssymb}
\usepackage{array}
\usepackage{booktabs}
\usepackage{amsthm}
\usepackage{url}
\usepackage{natbib}
\bibpunct{(}{)}{;}{a}{}{ }
\usepackage{hyperref}
\usepackage{bookmark}
\usepackage{authblk}
\usepackage{microtype}
\usepackage{enumitem}

\setlist[itemize]{itemsep=0.2em,topsep=0.2em,parsep=0em,partopsep=0em}
\setlist[enumerate]{itemsep=0.2em,topsep=0.2em,parsep=0em,partopsep=0em}
\theoremstyle{plain}
\newtheorem{theorem}{Theorem}[section]
\newtheorem{proposition}[theorem]{Proposition}
\newtheorem{lemma}[theorem]{Lemma}
\newtheorem{corollary}[theorem]{Corollary}

\theoremstyle{definition}
\newtheorem{definition}[theorem]{Definition}
\newtheorem{assumption}[theorem]{Assumption}
\newtheorem{example}[theorem]{Example}

\theoremstyle{remark}
\newtheorem{remark}[theorem]{Remark}

\providecommand{\keywords}[1]{\vspace{0.5em}\noindent\textbf{Keywords: }\small #1}

\newcommand{\PRS}{\ensuremath{\mathsf{PRS}}}
\newcommand{\BRK}{\ensuremath{\mathsf{BRK}}}
\newcommand{\BF}{\ensuremath{\mathsf{BF}}}
\newcommand{\NEU}{\ensuremath{\mathsf{NEU}}}
\newcommand{\NAcls}{\ensuremath{\mathsf{N/A}}}

\newcommand{\LOC}{\ensuremath{\mathrm{LOC}}}
\newcommand{\OBJ}{\ensuremath{\mathrm{OBJ}}}
\newcommand{\SCOPEE}{\ensuremath{\mathrm{SCOPE\text{-}E}}}
\newcommand{\SCOPES}{\ensuremath{\mathrm{SCOPE\text{-}S}}}
\newcommand{\RULEC}{\ensuremath{\mathrm{RULE\text{-}C}}}
\newcommand{\RULES}{\ensuremath{\mathrm{RULE\text{-}S}}}

\newcommand{\Fset}{\ensuremath{\mathcal{F}}}
\newcommand{\Class}{\ensuremath{\mathrm{Class}}}
\newcommand{\Sib}{\ensuremath{\mathrm{Sib}}}
\newcommand{\IT}{\ensuremath{\mathsf{IT}}}
\newcommand{\Ker}{\ensuremath{\mathrel{\approx_{d}}}}
\newcommand{\Dop}{\ensuremath{\Delta^{\mathrm{op}}_{d}}}
\newcommand{\Dsib}{\ensuremath{\Delta^{\mathrm{sib}}}}
\newcommand{\PASS}{\textnormal{\textsc{pass}}}
\newcommand{\FAIL}{\textnormal{\textsc{fail}}}
\newcommand{\INDET}{\textnormal{\textsc{indeterminate}}}

\title{
Operational Identity:\\
A Finite Audit of Declared and Implemented\\
Rules of Sameness
}

\author{Denise M. Case}
\affil{\small School of Computer Science and Information Systems\\
Northwest Missouri State University, Maryville, MO, USA}
\date{}

\begin{document}
\maketitle
\vspace{-2em}

\begin{abstract}
  A record system declares when two records refer to the same entity, occurrence, scope, or rule.
  Its disclosed implementation mechanisms induce a corresponding operational identity relation.
  The declared and implemented relations may diverge systematically
  without producing a provenance gap or detectable contradiction.
  A system can apply, consistently and with every record individually correct,
  a rule of sameness that no artifact declares.
  This paper formalizes that implemented relation.
  A declared identity regime partitions a finite record domain into co-reference classes;
  a disclosed mechanism, through its typed identity-relevant outcomes, induces an \emph{operational identity partition} of the same domain.
  The audit compares these partitions in the refinement lattice.
  A mechanism is \emph{faithful} when the declared partition refines the operational partition, so no declared class is split.
  A \emph{divergence witness} is a pair the declaration merges and the mechanism separates; such witnesses are decidable by pair enumeration.
  When an imported sibling basis also splits a declared class, local comparison with its partition yields sibling-aligned, sub-sibling, super-sibling, or sibling-incomparable divergence.
  This classification reports only the relationship; it does not identify the basis carried by the mechanism.
  Global equality of the operational and sibling partitions is defined separately as \emph{regime substitution} and does not follow from sibling alignment.
  A version field incremented on every textual edit inhabits the sub-sibling case by splitting declared classes more finely than either imported basis.
  The audit is three-valued and relative to the disclosed artifacts, evaluated surfaces, and identified uses; each boundary has a finite refuting witness.
  A passing verdict is non-monotone because extending the transformation history can merge declared classes and create a witness among records already examined.
\end{abstract}


\section{Introduction}
\label{oi.sec.intro}

Modern accountability systems increasingly operate under persistent disagreement.
Participants may contest what occurred, why it occurred,
which rules apply, who bears responsibility, or what consequences should follow.
Stable reference keeps competing claims anchored to the same subject matter throughout the dispute.

The record infrastructure supporting those systems must preserve that reference
as records are transformed (e.g. amended, forked, refined, decomposed, aggregated, reclassified)
and exchanged across institutional boundaries.
It must also distinguish the identity commitments needed for shared reference
from the causal, legal, and normative conclusions that remain open to dispute.
A record system that collapses those layers may make continued reference depend
on acceptance of the interpretation under examination.

\emph{Neutral Substrates}~\citep{case2026neutral}
establishes the neutrality-by-design constraint for this setting.
The foundational layer carries stable reference and permitted attribution,
while causal and normative interpretation remains in attributed extensions.
\emph{Referential Regimes}~\citep{case2026referential}
identifies regime-relative identity bases and transformation classifications
needed to preserve reference under ordinary change.
Together, the papers describe what a system should declare;
this paper examines what deployed systems do.

\subsection{Two Answers to One Question}

A record system answers the question
\emph{do these two records refer to the same thing?}
twice.

It answers once in its \emph{declaration}.
The declaration states what must remain fixed for a referent to remain the same,
and a transformation history then determines which records co-refer.
That answer partitions the records into co-reference classes.

It answers again in its \emph{behavior}.
An application assigns referent identifiers,
resolves references, merges or forks records,
admits or rejects operations, and selects transformation rules.
Those behaviors also sort records into
groups that the system treats as the same
and groups that it treats as different.
They partition the same records a second way.

The two partitions may differ.
When they do, the system operates under a rule of sameness
that no artifact states,
and may do so with perfect internal consistency.
A schema may declare content-fixed identity for a rule.
A workflow may decide that a structural revision creates a new rule.
An application may use a version field to refuse co-reference between two texts.
Each component behaves consistently.
The combined system applies an identity model that no single artifact declares.

This failure may be invisible to ordinary consistency checking.
A system that consistently applies an undeclared sameness rule
issues no contradictory claims.
Its provenance graph may record every revision accurately.
Its audit log may record every identifier assignment faithfully.
The missing fact is the relation between two partitions.

\subsection{The Operational Identity Partition}

This paper names the second partition and makes it computable.

A disclosed implementation mechanism is admitted as an \emph{audit surface}
when its value determines a typed, finite set of identity-relevant outcomes:
referent-identifier assignment, referent co-reference, referent persistence,
transformation classification, applicability selection, or operation admissibility.
The requirement that those outcomes be determined by the mechanism's value
is what licenses reading a difference in treatment as controlled by it.
The requirement that the mechanism be drawn from a disclosed registry
is what prevents the auditor from manufacturing surfaces at will.

Grouping records by their identity-relevant outcomes
yields the \emph{operational identity partition} induced by that surface.
The declared regime yields the declared partition.
Both are partitions of the same finite record domain,
and the audit is a comparison of the two
in the refinement lattice.

Faithfulness is one direction of that comparison:
the declared partition refines the operational one,
so the surface never splits a declared co-reference class.
Its refutation is a finite object.
A \emph{divergence witness} is a pair \((r,s)\) with
\[
  r\sim_{\tau}s
  \qquad\text{and}\qquad
  \IT_d(r)\neq\IT_d(s).
\]
The declaration merges the pair.
The surface separates it.
The surface therefore carries an identity distinction the declaration does not.

\subsection{Classifying the Divergence}

Six of the nine imported regimes lie on a sibling axis,
where a second identity basis is available for the same carrier kind.
When that sibling also splits at least one examined declared class,
the divergence can be classified.
Restricting both the operational relation and the sibling relation
to the declared co-reference classes yields two equivalence relations
that fall into one of four cases:
aligned with the sibling,
strictly finer than the sibling,
strictly coarser than the sibling,
or incomparable with the sibling.

The comparison is local to the declared classes,
and that locality bounds what alignment can establish.
Sibling alignment says the surface splits declared classes as the sibling does.
It does not say the surface implements the sibling regime,
because the two may still disagree on pairs the declaration separates.
The stronger global condition,
that the operational partition equals the sibling partition on the whole domain,
is defined separately as \emph{regime substitution},
and Proposition~\ref{oi.prop.aligned-not-substitution}
shows sibling alignment does not imply it.
Regime substitution is the formal condition under which
a surface implements a hidden identity regime.

The worked version-field example occupies
the sub-sibling case rather than the aligned one.
A version field incremented on every textual edit
splits declared co-reference classes strictly more finely than the sibling.
Section~\ref{oi.sec.worked} exhibits that case,
and Proposition~\ref{oi.prop.witness-not-alignment}
proves it cannot be told from sibling alignment by any single witness pair.
Diagnosing it as regime substitution would prescribe the wrong repair.

The audit extends to the three regimes with no sibling.
Faithfulness and divergence are defined for every regime in the inventory.
A divergence is unpositioned when no sibling is available
or when the sibling makes no additional distinction
within the examined declared classes.

\subsection{Contributions}

This paper makes five contributions.

First, it defines the operational identity partition,
the rule of sameness induced
by an examined implementation surface,
as a computable object over a disclosed mechanism registry.

Second, it defines faithfulness as a refinement relation
between the declared and operational partitions,
gives the surface-level asymmetry under which a split refutes faithfulness
while a merge does not,
and supplies a finite divergence witness that refutes it.

Third, where the sibling is informative on the examined domain,
it classifies divergence against the sibling partition
in four local lattice cases,
distinguishes that classification from the global condition
under which the surface actually implements the sibling regime,
and gives a finite version-field construction
occupying the sub-sibling case.

Fourth, it proves that a passing verdict is not monotone
under extension of the transformation history alone.
Because declared co-reference is a transitive closure,
adding history to a fixed record domain
can merge previously separate declared classes
and create a witness among records the audit already examined,
with no change to the surface and no change to those records.

Fifth, it makes explicit that the audit is disclosure-relative in three respects.
The artifacts disclosed,
the surfaces evaluated over them,
and the identity-relevant uses identified for each
are all disclosure artifacts,
each carries a completeness claim the procedure cannot discharge,
and each has a witness object whose later discovery refutes that claim.

The paper introduces no record architecture and requires no reference checker.
An auditor may construct the required finite inputs
from provenance records, schemas, workflows, configuration,
source inspection, runtime traces, or differential tests.

Section~\ref{oi.sec.interface} states the imported interface and the preconditions.
Section~\ref{oi.sec.surface} defines surfaces and the operational identity partition.
Section~\ref{oi.sec.compare} compares the two partitions.
Section~\ref{oi.sec.procedure} gives decidability, the verdict, and non-monotonicity.
Section~\ref{oi.sec.worked} works a legal-alignment example.
Section~\ref{oi.sec.disclosure} states the disclosure boundary.
Section~\ref{oi.sec.related} situates the result.
Section~\ref{oi.sec.limits} states its limits, and
Section~\ref{oi.sec.conclusion} concludes.

\section{Imported Interface and Preconditions}
\label{oi.sec.interface}

This paper uses a small interface imported from
\emph{Neutral Substrates} and \emph{Referential Regimes}.
The earlier papers supply the neutrality argument,
carrier inventory, transformation basis,
identity-basis construction, regime assignments,
and nine-regime lower bound.
The present result needs regime-relative classification,
induced co-reference, and the three sibling axes.

\subsection{Regimes and Declared Co-reference}

\begin{definition}[Identity Regime]
  \label{oi.def.regime}
  An \emph{identity regime} \(\rho\) consists of an identity basis
  and a classification function
  \[
    \Class_{\rho}:\Fset\longrightarrow
    \{\PRS,\BRK,\NEU,\NAcls\}.
  \]
  The identity basis states what must remain fixed
  for a referent to remain the same.
  The classification function states whether each transformation family
  preserves identity, breaks identity, is identity-neutral,
  or is inapplicable under that basis.
\end{definition}

Let a finite family-labeled transformation history be
\[
  H=\langle(f_i,x_i,y_i)\rangle_{i=1}^{n},
\]
where \(f_i\in\Fset\) is the transformation family
and \(x_i,y_i\in R\) are its source and result records.
The family label is regime-independent.
The classification changes with the regime.
For histories \(H\) and \(H^{\ast}\), write \(H\preceq H^{\ast}\)
when \(H\) is a subsequence of \(H^{\ast}\),
so \(H^{\ast}\) extends \(H\) without removing or changing
any transformation already recorded.

\begin{definition}[Non-breaking Edge Set]
  \label{oi.def.nonbreaking}
  For a regime \(\rho\), define
  \[
    E_{\rho}
    =
    \{
    (x_i,y_i)
    \mid
    \Class_{\rho}(f_i)\in\{\PRS,\NEU\}
    \}.
  \]
\end{definition}

\begin{remark}[Why Neutral Families Contribute Edges]
  \label{oi.rem.neutral}
  A family classified \PRS{} under \(\rho\)
  acts on the \(\rho\)-identity basis and leaves it fixed.
  A family classified \NEU{} under \(\rho\) does not act on that basis at all.
  Neither can separate referents under \(\rho\),
  so both yield non-breaking edges,
  for different reasons.
  Only \BRK{} severs the relation.
  A family classified \NAcls{} does not apply to the carrier kind,
  and Definition~\ref{oi.def.precond} excludes it
  from a well-formed history for records of that kind.
\end{remark}

\begin{definition}[Declared Identity Relation and Partition]
  \label{oi.def.declared}
  The declared identity relation \(\sim_{\rho}\)
  is the equivalence relation generated by \(E_{\rho}\).
  Its quotient
  \[
    \pi_{\rho}
    =
    R\,/\!\sim_{\rho}
  \]
  is the \emph{declared identity partition} of the record domain \(R\).
  Each block is a set of records the regime treats
  as representations of one referent.
\end{definition}

\begin{definition}[Core Sibling Axes]
  \label{oi.def.axes}
  The core sibling axes are
  \[
    \mathcal{A}
    =
    \bigl\{
    \{\LOC,\OBJ\},
    \{\SCOPEE,\SCOPES\},
    \{\RULEC,\RULES\}
    \bigr\}.
  \]
  Each pair supplies two imported identity bases for one carrier kind:
  locus-fixed and object-fixed identity for plain referents,
  extension-fixed and structure-fixed identity for scopes,
  and content-fixed and structure-fixed identity for rules.
\end{definition}

\begin{definition}[Sibling Function]
  \label{oi.def.sibling}
  For a regime \(\tau\), define
  \[
    \Sib(\tau)
    =
    \begin{cases}
      \tau',
       & \text{when }\{\tau,\tau'\}\in\mathcal{A};
      \\
      \varnothing,
       & \text{when \(\tau\) lies on no imported sibling axis.}
    \end{cases}
  \]
\end{definition}

\begin{definition}[Imported Referential-Regime Interface]
  \label{oi.def.imported-interface}
  The \emph{imported referential-regime interface} is
  \[
    \mathfrak{R}
    =
    \bigl(
    \Fset,
    \mathcal{P},
    \{\Class_{\rho}\}_{\rho\in\mathcal{P}},
    \Sib
    \bigr),
  \]
  where:
  \begin{enumerate}[label=(\alph*)]
    \item \(\Fset\) is the finite imported transformation-family set;
    \item \(\mathcal{P}\) is the finite imported regime inventory;
    \item each
          \[
            \Class_{\rho}:
            \Fset\longrightarrow
            \{\PRS,\BRK,\NEU,\NAcls\}
          \]
          is decidable; and
    \item \(\Sib\) is the imported partial sibling function.
  \end{enumerate}
\end{definition}

\begin{remark}[Interface Revision]
  \label{oi.rem.interface-revision}
  Every declared partition,
  sibling partition,
  divergence classification,
  and audit verdict in this paper
  is relative to a fixed imported interface \(\mathfrak{R}\).

  Adding a transformation family,
  adding an identity regime,
  revising a classification function,
  or revising the sibling relation
  produces a new interface version.
  The declared and sibling partitions must then be recomputed,
  the finite audit must be rerun,
  and inventory-relative examples and claims must be rechecked.

  The general definitions of operational identity,
  faithfulness,
  divergence,
  store composition,
  and finite decidability remain applicable
  when the revised interface remains finite,
  carrier-consistent,
  and decidable.
\end{remark}

Three regimes in the imported inventory lie on no sibling axis.
The audit of Section~\ref{oi.sec.compare} decides faithfulness for all nine.
Only the diagnosis of a failure uses \(\Sib(\tau)\),
and Section~\ref{oi.sec.compare} states what is returned
when \(\Sib(\tau)=\varnothing\).

\begin{remark}[Siblings Need Not Refine One Another]
  \label{oi.rem.norefine}
  Neither declared relation on a sibling axis refines the other in general.
  A structural revision preserving rule content
  is non-breaking under \RULEC{} and breaking under \RULES{}.
  A wording change altering content while preserving section structure
  is breaking under \RULEC{} and non-breaking under \RULES{}.
  The two relations cross.
  The comparison of Section~\ref{oi.sec.compare}
  is therefore carried out inside declared co-reference classes,
  where a containment can be asked for,
  rather than globally, where none may be assumed.
\end{remark}

\subsection{Preconditions on the Audit Artifact}

The procedure consumes a declared regime assignment and a family-labeled history.
It does not establish that either is correct.

\begin{definition}[Audit Preconditions]
  \label{oi.def.precond}
  For an examined record domain \(R\)
  of a single carrier kind with declared regime \(\tau\):
  \begin{enumerate}[label=(P\arabic*)]
    \item \emph{Regime determinacy.}
          Each record in \(R\) has exactly one effective declared regime,
          and that regime is \(\tau\).
    \item \emph{History well-formedness.}
          Every transformation in \(H\) acts on the carrier kind of \(R\),
          so \(\Class_{\rho}(f_i)\neq\NAcls\) for every \(i\)
          and every regime \(\rho\) assigned to that kind.
    \item \emph{History-label validity.}
          Each family label \(f_i\) accurately names
          the transformation performed.
          The declared identity effect of that transformation
          is obtained by applying
          \(\Class_{\tau}\) to \(f_i\).
          Any identity effect applied by the implementation
          is part of the identity-relevant treatment being audited
          and is not assumed to agree with the declaration.
  \end{enumerate}
\end{definition}

\begin{remark}[Preconditions Are Inputs]
  \label{oi.rem.precond}
  The preconditions are assumptions on the audit artifact.
  The declared partition is computed
  from the regime assignment and the recorded family labels.
  An indeterminate regime assignment,
  an unresolved endpoint,
  an unavailable family label,
  or a family label that does not accurately name the transformation performed
  makes the declared partition unreliable.
  A comparison against an unreliable declared partition
  carries no information about the surface.
  Failure to discharge one of those conditions yields \INDET{},
  not \PASS{} or \FAIL{}.

  A difference between the identity effect
  prescribed by \(\Class_{\tau}(f_i)\)
  and the identity effect applied by the implementation
  is not a precondition failure.
  It is a candidate operational-identity divergence.
\end{remark}

\section{Surfaces and the Operational Identity Partition}
\label{oi.sec.surface}

Identity behavior may be implemented outside a formal regime declaration.
A role may select an identifier policy.
A workflow state may determine whether two revisions co-refer.
A schema convention may change persistence handling.
An application predicate may route records through different transformation rules.

The presence of such a mechanism is not a failure.
A mechanism may serve functions unrelated to identity.
The audit must first say what makes a mechanism identity-bearing.

\subsection{The Surface Registry}

\begin{definition}[Mechanism Registry]
  \label{oi.def.registry}
  Let \(B\) be an identified implementation boundary
  and \(R\) the examined record domain.
  A \emph{mechanism registry} \(D_B\) is a finite set
  of implementation artifacts within \(B\).

  Each \(d\in D_B\) is an identifiable field,
  workflow state,
  schema convention,
  configuration value,
  application predicate,
  or policy input,
  together with a total computable observation map
  \[
    d:R\longrightarrow V_d.
  \]
  A missing or inapplicable value is represented
  by a distinguished value in \(V_d\).

  Each artifact is disclosed by the operator
  or independently discovered by the auditor.
\end{definition}

\begin{remark}[Why the Registry Is Necessary]
  \label{oi.rem.registry}
  Every function on the record domain induces a partition of it.
  Without the registry constraint,
  an auditor could take the identity function on records as a mechanism,
  which separates every distinct pair
  and therefore reports a divergence in any system
  that does anything at all with two co-referring records.
  The registry is what makes a surface an artifact of the implementation
  rather than an artifact of the audit.
  A candidate mechanism must be exhibitable:
  the auditor must be able to identify the implementation artifact that carries it.
\end{remark}

\begin{definition}[Joint Mechanism]
  \label{oi.def.joint}
  For a nonempty finite \(S\subseteq D_B\), the \emph{joint mechanism} \(d_S\) is
  \[
    d_S(r)
    =
    \bigl\langle d(r) \;\bigm|\; d\in S \bigr\rangle,
  \]
  under a fixed order on \(S\).
  Write \(\mathcal{D}_B\) for the set of joint mechanisms over nonempty finite subsets of \(D_B\),
  which contains each \(d\in D_B\) as the singleton case.
\end{definition}

The auditor must exhibit the constituent mechanisms of a joint surface,
so \(\mathcal{D}_B\) is finite,
every element of it is an exhibitable composite of registered artifacts,
and the identity function on records is not among them
unless \(D_B\) already separates every pair of records,
in which case the implementation has itself disclosed a mechanism that does so.

\subsection{Identity-Relevant Outcomes}

\begin{definition}[Identity-Relevant Outcome Types]
  \label{oi.def.outcomes}
  An outcome is \emph{identity-relevant}
  when it is of one of the following types:
  \begin{enumerate}[label=(\alph*)]
    \item \emph{referent-identifier assignment:}
          which referent identifier a record is bound to;
    \item \emph{referent resolution:}
          the referent identifier,
          referent key,
          or co-reference class
          to which the system resolves a record;
    \item \emph{referent persistence:}
          whether a referent is continued, retired, or created;
    \item \emph{transformation treatment:}
          a finite record-indexed encoding
          of the transformation family
          or identity effect applied to the record
          within the examined history;
    \item \emph{applicability selection:}
          which identity rule or basis governs a record; or
    \item \emph{operation admissibility:}
          whether an operation preserving or breaking identity is permitted.
  \end{enumerate}
\end{definition}

\begin{remark}[Scope of the Outcome Vocabulary]
  \label{oi.rem.outcome-vocabulary}
  Definition~\ref{oi.def.outcomes}
  fixes the identity-relevant outcome vocabulary for the present audit.
  The enumeration leaves open whether other audit specifications
  may admit additional identity-relevant outcome types.
  A revision of the finite outcome vocabulary produces a revised audit specification,
  under which the identified uses, identity-treatment signatures,
  and operational partitions must be recomputed.
\end{remark}

\begin{remark}[Record Identifiers Are Not Identity-Relevant]
  \label{oi.rem.recordid}
  Every record in a well-formed store carries a distinct record identifier by construction.
  A record identifier individuates the representation, rather than the referent.
  Two records co-referring under any regime therefore differ in record identifier,
  and admitting that difference as an identity-relevant outcome
  would produce a divergence witness in every system, including every correct one.
  The outcome types of Definition~\ref{oi.def.outcomes}
  are stated over referents throughout.
  The distinction between the two identifier layers
  is a precondition on the audit vocabulary, rather than a matter of exposition.
\end{remark}

\begin{remark}[Record-Indexed Encoding]
  \label{oi.rem.record-indexed}
  The operational partition is induced by total record-indexed outcome functions.

  A pairwise co-reference procedure may be represented,
  over finite \(R\), by the vector of its decisions against the records of \(R\)
  under a fixed order.
  A transformation-level procedure may be represented by a fixed finite record-indexed encoding
  of the identity-relevant transformation outcomes incident to each record.

  This encoding is part of the surface's identified use:
  it must be a total function of the record
  factoring through the mechanism value
  in the sense of Definition~\ref{oi.def.surface}(c),
  and the operational partition
  induced by a transformation-treatment outcome
  is defined relative to it.
  Two auditors who fix different such encodings
  evaluate different uses,
  and the disclosure-relativity of
  Section~\ref{oi.sec.disclosure} already ranges over that choice.

  When no finite total record-indexed encoding is available,
  the procedure does not induce the operational partition defined here,
  and the verdict is \INDET{}.
\end{remark}

\subsection{Surfaces}

\begin{definition}[Audit Surface]
  \label{oi.def.surface}
  An \emph{audit surface} over a record domain \(R\)
  of one carrier kind with declared regime \(\tau\) consists of:
  \begin{enumerate}[label=(\alph*)]
    \item a mechanism \(d\in\mathcal{D}_B\) with a computable value function
          \[
            d:R\longrightarrow V_d;
          \]
    \item a finite sequence of \emph{identified} identity-relevant uses
          \[
            U_d
            =
            \langle u_1,\ldots,u_m\rangle,
            \qquad
            u_j:R\longrightarrow O_j,
          \]
          each of a type drawn from Definition~\ref{oi.def.outcomes},
          with finite, computable outcome for every \(r\in R\).
          The sequence records the uses the audit has found,
          and Section~\ref{oi.sec.disclosure} states the completeness claim
          under which it may be treated as all of them;
    \item \emph{surface control:} each use factors through the mechanism value,
          so that for every \(j\) there exists
          \[
            \tilde{u}_j:V_d\longrightarrow O_j
            \qquad\text{with}\qquad
            u_j=\tilde{u}_j\circ d
            \quad\text{on }R.
          \]
  \end{enumerate}
\end{definition}

\begin{remark}[Surface Granularity]
  \label{oi.rem.granularity}
  Clause (c) is an admissibility condition on how the auditor identifies the mechanism.
  When the identity-relevant uses of a candidate mechanism
  do not factor through its value,
  the outcome is jointly controlled by that mechanism and others,
  and the auditor escalates to the joint mechanism \(d_S\)
  of Definition~\ref{oi.def.joint}
  for a subset \(S\subseteq D_B\) that together determines the outcome.
  The escalation ranges over \(\mathcal{D}_B\) and terminates,
  because \(D_B\) is finite.
  It cannot terminate in an arbitrary function of the record,
  which Remark~\ref{oi.rem.registry} excludes.
  When no \(S\subseteq D_B\) yields factorization,
  the audit returns \INDET{}:
  the disclosed registry does not contain the mechanisms
  that determine the outcome.
\end{remark}

\begin{definition}[Examined Surface Family]
  \label{oi.def.family}
  An \emph{examined surface family} \(A_B\)
  is a finite set of audit surfaces over \(\mathcal{D}_B\),
  each supplied with its own identified-use sequence.
  The family is what the audit actually evaluates.
  It contains the singleton surfaces the auditor admitted
  and any joint surfaces of Definition~\ref{oi.def.joint}
  to which granularity escalation carried them.
\end{definition}

The distinction between \(D_B\), \(\mathcal{D}_B\), and \(A_B\) does work.
\(D_B\) is the disclosed inventory of implementation artifacts.
\(\mathcal{D}_B\) is its closure under joint value,
which is what makes surface control attainable.
\(A_B\) is the finite subfamily the audit evaluated,
and it is the family over which store-level statements are made.
An outcome jointly controlled by two registered fields
and factoring through neither alone
appears in \(A_B\) as a joint surface and nowhere else,
so store-level conclusions must range over \(A_B\)
rather than over \(D_B\).

\begin{definition}[Identity-Treatment Signature and Operational Partition]
  \label{oi.def.signature}
  For an audit surface \(d\) and record \(r\in R\), define
  \[
    \IT_d(r)
    =
    \langle
    u_1(r),\ldots,u_m(r)
    \rangle.
  \]
  Write \(r\Ker s\) when \(\IT_d(r)=\IT_d(s)\).
  The relation \(\Ker\) is an equivalence relation on \(R\),
  and its quotient
  \[
    \pi_d
    =
    R\,/\Ker
  \]
  is the \emph{operational identity partition} induced by \(d\).
\end{definition}

The blocks of \(\pi_d\) are the sets of records
to which the implementation gives the same identity-relevant treatment
through the identified uses of \(d\).
It is derived from behavior,
and it is available to an auditor
whether or not the system's declaration mentions it.
It is the rule of sameness the mechanism applies
precisely to the extent that \(U_d\) is complete,
which is a disclosure claim and not a computation.
Definition~\ref{oi.def.usecomplete} states that claim,
and Definition~\ref{oi.def.omitteduse} states the object that refutes it.

\begin{lemma}[Treatment Is Determined by Mechanism Value]
  \label{oi.lem.determined}
  Let \(d\) be an audit surface and \(r,s\in R\).
  If \(d(r)=d(s)\) then \(r\Ker s\).
  Equivalently, \(\IT_d(r)\neq\IT_d(s)\) implies \(d(r)\neq d(s)\).
\end{lemma}

\begin{proof}
  By surface control, \(u_j=\tilde{u}_j\circ d\) for each \(j\).
  If \(d(r)=d(s)\) then
  \(u_j(r)=\tilde{u}_j(d(r))=\tilde{u}_j(d(s))=u_j(s)\)
  for every \(j\),
  so the signatures agree.
  The second statement is the contrapositive.
\end{proof}

Lemma~\ref{oi.lem.determined} says that \(\pi_d\) is coarser than
the partition of \(R\) by mechanism value.
A difference in identity treatment entails a difference in mechanism value,
so an auditor who exhibits the former has already exhibited the latter.
The witness of Section~\ref{oi.sec.compare}
therefore needs no separate clause on mechanism values,
and reports them as evidence rather than requiring them as a condition.

The same fact disposes of the boundary case in
Definition~\ref{oi.def.joint}.
A joint mechanism whose value separates every pair of records,
that is, the case in which the registry is itself injective,
still induces \(\pi_d\) through its uses,
and those uses may merge records the value distinguishes.
Value injectivity therefore splits no declared class on its own,
and excluding the record-identity function (Remark~\ref{oi.rem.registry})
costs the auditor nothing on a system
that has genuinely disclosed a pair-separating mechanism.

\section{Comparing Declared and Operational Identity}
\label{oi.sec.compare}

Both \(\pi_{\tau}\) and \(\pi_d\) are partitions of the same finite domain \(R\).
The audit is their comparison in the refinement lattice.
Throughout, a partition \(\pi\) \emph{refines} \(\sigma\)
when every block of \(\pi\) is contained in a block of \(\sigma\),
equivalently when \(\sim_{\pi}\,\subseteq\,\sim_{\sigma}\) as relations.

\subsection{Faithfulness}

\begin{definition}[Faithful Surface]
  \label{oi.def.faithful}
  An audit surface \(d\) is \emph{faithful} to the declared regime \(\tau\) on \(R\)
  when \(\pi_{\tau}\) refines \(\pi_d\), that is, when
  \[
    r\sim_{\tau}s
    \;\Longrightarrow\;
    r\Ker s
    \qquad\text{for all }r,s\in R.
  \]
\end{definition}

A faithful surface never splits a declared co-reference class.
It may merge them.

\begin{remark}[Why Faithfulness Is One-Directional]
  \label{oi.rem.asymmetry}
  The two directions of the refinement relation are not symmetric faults,
  and the asymmetry follows from what a single surface is.

  A surface controls some identity-relevant outcomes.
  It does not carry the system's whole identity assignment.
  When a surface merges two records the declaration separates,
  it reports only that this mechanism does not itself
  distinguish those referents,
  and another mechanism may.
  A version field that is identical across two unrelated rules
  has told the audit nothing false.

  When a surface splits a declared co-reference class,
  the system has already given the two records
  different referent identifiers,
  refused to resolve them together,
  or applied different persistence handling.
  No further mechanism can retract that treatment.
  The distinction has been made.

  Conflation of declared-distinct referents is therefore
  a property of the store rather than of a surface,
  and it lies outside the scope of a single-surface audit.
  What a surface can be held to
  is that it introduces no distinction the declaration lacks.
\end{remark}

\begin{definition}[Divergence Witness]
  \label{oi.def.divergence}
  A pair \((r,s)\in R\times R\)
  is a \emph{divergence witness} for \((d,\tau)\) when
  \[
    r\sim_{\tau}s
    \qquad\text{and}\qquad
    \IT_d(r)\neq\IT_d(s).
  \]
\end{definition}

\begin{proposition}[Divergence Refutes Faithfulness]
  \label{oi.prop.refutes}
  A surface \(d\) is faithful to \(\tau\) on \(R\)
  if and only if no divergence witness for \((d,\tau)\) exists in \(R\times R\).
  Every divergence witness satisfies \(d(r)\neq d(s)\).
\end{proposition}

\begin{proof}
  Definition~\ref{oi.def.faithful} is a universally quantified implication
  whose negation at a pair is exactly Definition~\ref{oi.def.divergence}.
  The second statement is Lemma~\ref{oi.lem.determined}.
\end{proof}

A divergence witness is an audit finding on its own.
It shows that identity treatment in the deployed system
depends on a distinction the declared regime does not make.
It is available for every regime in the inventory,
including the three that lie on no sibling axis.
It does not yet say which distinction.

\subsection{Faithfulness Composes}

A single surface carries part of the system's identity treatment.
The examined identity treatment is carried by \(A_B\).

\begin{definition}[Store Operational Partition]
  \label{oi.def.store}
  For an examined surface family \(A_B\) on \(R\), define
  \[
    \approx_{B}
    \;=\;
    \bigcap_{d\in A_B}\;\Ker,
    \qquad
    \pi_B
    =
    R\,/\approx_{B}.
  \]
  Two records are operationally the same in the store
  when every examined surface,
  singleton or joint,
  gives them the same identity-relevant treatment.
\end{definition}

\begin{proposition}[Store Faithfulness Is Surface Faithfulness]
  \label{oi.prop.compose}
  The declared partition \(\pi_{\tau}\) refines \(\pi_B\)
  if and only if every \(d\in A_B\) is faithful to \(\tau\) on \(R\).
  A divergence witness for any surface in \(A_B\)
  is a divergence witness for the store.
\end{proposition}

\begin{proof}
  \(\approx_{B}\) is an intersection over \(A_B\),
  so \(\sim_{\tau}\subseteq\;\approx_{B}\)
  holds exactly when \(\sim_{\tau}\subseteq\;\Ker\) for every \(d\in A_B\).
  A pair witnessing the failure of one conjunct witnesses the failure of the whole.
\end{proof}

Faithfulness therefore composes.
A divergence on any surface refutes store faithfulness,
while the store passes only when every surface in \(A_B\) is faithful.
The audit may proceed one surface at a time
while retaining those quantified conclusions.

The quantification over \(A_B\) rather than \(D_B\) is not cosmetic.
An outcome jointly controlled by two registered fields, factoring through neither alone, is carried by a joint surface.
That surface may split a declared co-reference class
while each constituent field, taken alone, splits nothing.
Intersecting only over the singleton mechanisms of \(D_B\)
would omit the surface that carries the fault,
and Proposition~\ref{oi.prop.compose} would then be false.

Under the completeness claims of Section~\ref{oi.sec.disclosure},
\(\pi_B\) is the rule of sameness the system implements.
Proposition~\ref{oi.prop.compose} licenses a store-level failure
from any failing surface and a store-level pass
only from faithfulness of every surface in \(A_B\).

\subsection{The Lattice Comparison}

Diagnosis restricts both relations to the declared co-reference classes.

\begin{definition}[Restricted Relations]
  \label{oi.def.restricted}
  Let \(\tau'=\Sib(\tau)\neq\varnothing\).
  Define the equivalence relations
  \[
    \Dop
    =
    \;\Ker\;\cap\;\sim_{\tau},
    \qquad
    \Dsib
    =
    \;\sim_{\tau'}\;\cap\;\sim_{\tau}.
  \]
\end{definition}

Both are equivalence relations contained in \(\sim_{\tau}\).
\(\Dop\) records which pairs the declaration merges
and the surface also merges.
\(\Dsib\) records which pairs the declaration merges
and the sibling basis also merges.
Restricting to \(\sim_{\tau}\) is what makes the comparison well posed,
given Remark~\ref{oi.rem.norefine}.

Faithfulness is now the statement \(\Dop=\;\sim_{\tau}\).
When it fails, \(\Dop\) is a proper subrelation of \(\sim_{\tau}\).
If \(\Dsib\) is also a proper subrelation of \(\sim_{\tau}\),
the two restricted relations fall into one of four cases.
If \(\Dsib=\;\sim_{\tau}\),
the divergence is unpositioned.

\begin{definition}[Divergence Classification]
  \label{oi.def.classification}
  Suppose \(d\) is not faithful to \(\tau\) on \(R\),
  that \(\Sib(\tau)=\tau'\neq\varnothing\),
  and that \(\Dsib\neq\;\sim_{\tau}\),
  so the sibling basis also splits some declared class.
  Exactly one of the following holds:
  \begin{enumerate}[label=(\alph*)]
    \item \(\Dop=\Dsib\):
          \emph{sibling-aligned divergence}.
          Inside the declared classes,
          the surface splits exactly as the sibling does.
    \item \(\Dop\subsetneq\Dsib\):
          \emph{sub-sibling divergence}.
          Inside the declared classes,
          the surface splits strictly more finely than the sibling.
    \item \(\Dsib\subsetneq\Dop\):
          \emph{super-sibling divergence}.
          Inside the declared classes,
          the surface makes some but not all of the sibling's distinctions.
    \item \(\Dop\) and \(\Dsib\) are incomparable:
          \emph{sibling-incomparable divergence}.
          The surface makes distinctions the sibling does not
          and omits distinctions the sibling makes.
  \end{enumerate}
  When \(\Sib(\tau)=\varnothing\), or when \(\Dsib=\;\sim_{\tau}\),
  the divergence is \emph{unpositioned}:
  the surface splits a declared class,
  and the imported inventory supplies no sibling distinction on \(R\)
  against which the split can be located.
\end{definition}

Under the hypotheses of Definition~\ref{oi.def.classification},
the four cases are exhaustive and mutually exclusive,
being the four possible relative positions
of two elements of a partially ordered set
under the containment order.

\begin{remark}[What the Classification Does Not Say]
  \label{oi.rem.local}
  Definition~\ref{oi.def.classification} compares
  \(\Dop\) and \(\Dsib\),
  both of which are intersected with \(\sim_{\tau}\).
  The comparison is therefore local to the declared co-reference classes,
  and it is silent about pairs the declaration separates.
  Sibling-aligned divergence establishes that
  the surface splits declared classes as the sibling would.
  It does not establish that
  \(\Ker\;=\;\sim_{\tau'}\) on \(R\),
  because the surface may treat a \(\tau\)-separated pair
  in a way the sibling does not.
  The names of the four cases are relational for that reason,
  and none of them asserts that the surface implements a named basis.
\end{remark}

\begin{definition}[Regime Substitution]
  \label{oi.def.regime-substitution}
  A surface \(d\) exhibits \emph{regime substitution} on \(R\)
  when \(\Sib(\tau)=\tau'\neq\varnothing\),
  the regimes differ on \(R\),
  and the operational partition equals the sibling partition,
  \[
    \pi_d
    =
    \pi_{\tau'},
    \qquad\text{equivalently}\qquad
    \Ker\;=\;\sim_{\tau'}
    \text{ on }R.
  \]
  A surface exhibiting regime substitution
  implements a \emph{hidden identity regime}:
  a consistently applied, undeclared rule of sameness
  that the imported inventory already names.
\end{definition}

\begin{example}[A Forked Continuation on the Locus/Object Axis]
  \label{oi.ex.loc-obj-fork}
  Let \(r_0\) record a monitoring station at site \(\ell\) using sensor \(o_0\).
  A branch/fork creates two continuation records \(r_1\) and \(r_2\)
  for monitoring at the same site \(\ell\),
  using sensor objects \(o_1\) and \(o_2\), with \(o_0\), \(o_1\), and \(o_2\) carrying distinct serials.
  Let
  \[
    R=\{r_0,r_1,r_2\}
  \]
  and
  \[
    H
    =
    \bigl\langle
    (\BF{},r_0,r_1),
    (\BF{},r_0,r_2)
    \bigr\rangle.
  \]

  In the imported referential-regime interface,
  \[
    \Class_{\LOC}(\BF{})=\PRS
    \qquad\text{and}\qquad
    \Class_{\OBJ}(\BF{})=\BRK.
  \]
  Both fork edges are therefore non-breaking under \LOC{},
  giving
  \[
    \pi_{\LOC}
    =
    \bigl\{\{r_0,r_1,r_2\}\bigr\}.
  \]
  Both are breaking under \OBJ{},
  giving
  \[
    \pi_{\OBJ}
    =
    \bigl\{\{r_0\},\{r_1\},\{r_2\}\bigr\}.
  \]

  Suppose the declared regime is \LOC{}.
  Let \(d\) be the registered sensor-serial field,
  with a referent-identifier use that assigns
  a distinct station identifier to each sensor serial,
  so the station's referent identifier follows the sensor object.
  Then
  \[
    \pi_d=\pi_{\OBJ}.
  \]
  The pair \((r_0,r_1)\) is a divergence witness,
  because the declaration merges it
  and the surface separates it.
  The divergence is sibling-aligned,
  and the equality \(\pi_d=\pi_{\OBJ}\)
  establishes regime substitution.
\end{example}

Regime substitution is a global equality condition on \(R\).
It is distinct from sibling alignment, which compares the relations
only within declared identity classes.
The next subsection shows that neither sibling alignment
nor a single divergence witness establishes it.

\subsection{Neither a Witness nor an Alignment Diagnoses}

\begin{proposition}[Substitution Implies Sibling Alignment and a Witness]
  \label{oi.prop.substitution-witness}
  Suppose \(d\) exhibits regime substitution on \(R\)
  and that
  \[
    \Dsib
    \;\neq\;
    \sim_{\tau},
    \qquad\text{equivalently}\qquad
    \sim_{\tau}\;\not\subseteq\;\sim_{\tau'}.
  \]
  Then \(d\) is not faithful to \(\tau\) on \(R\),
  it exhibits sibling-aligned divergence,
  and any pair \((r,s)\) with \(r\sim_{\tau}s\) and \(r\not\sim_{\tau'}s\)
  is a divergence witness for \((d,\tau)\)
  that separates under the sibling.
\end{proposition}

\begin{proof}
  Substitution gives \(\Ker\;=\;\sim_{\tau'}\),
  so intersecting both sides with \(\sim_{\tau}\) gives \(\Dop=\Dsib\).
  The hypothesis \(\Dsib\neq\;\sim_{\tau}\)
  supplies a pair with \(r\sim_{\tau}s\) and \(r\not\sim_{\tau'}s\),
  which satisfies \(\IT_d(r)\neq\IT_d(s)\)
  because \(\Ker\;=\;\sim_{\tau'}\).
  That pair is a divergence witness,
  so faithfulness fails,
  and \(\Dop=\Dsib\) with \(\Dsib\neq\;\sim_{\tau}\)
  is sibling-aligned divergence.
\end{proof}

\begin{remark}[Regime Substitution Need Not Produce a Splitting Fault]
  \label{oi.rem.merging-substitution}
  The hypothesis of Proposition~\ref{oi.prop.substitution-witness}
  is not implied by the requirement that the regimes differ on \(R\),
  and dropping it makes the proposition false.

  Suppose the examined history contains only a content amendment
  preserving section structure,
  which is \BRK{} under \RULEC{} and \PRS{} under \RULES{}.
  Then \(\sim_{\RULEC}\) separates the two records
  and \(\sim_{\RULES}\) merges them,
  so the regimes differ on \(R\)
  while \(\sim_{\RULEC}\;\subseteq\;\sim_{\RULES}\).
  A surface with \(\pi_d=\pi_{\RULES}\)
  exhibits regime substitution:
  it globally implements the undeclared sibling basis.
  It nevertheless splits no declared co-reference class,
  so by Remark~\ref{oi.rem.asymmetry}
  it is faithful and the audit returns \PASS{}.

  The surface is running an identity regime
  different from the declaration,
  and it is running it by merging referents the declaration separates,
  which is the direction this audit does not police.
  The observation is a limit on the audit rather than a defect in the definitions.
  A splitting fault is refutable by one pair
  because no later mechanism can retract a distinction already made.
  A merging fault is not,
  because another mechanism in \(A_B\) may carry the distinction,
  and Definition~\ref{oi.def.store} is where a store-level
  merging analysis would have to begin.
\end{remark}

\begin{proposition}[Sibling Alignment Does Not Imply Substitution]
  \label{oi.prop.aligned-not-substitution}
  There is a carrier kind, a declared regime \(\tau\) with sibling \(\tau'\),
  a family-labeled history over the imported transformation basis,
  and a surface \(d\) exhibiting sibling-aligned divergence
  and not exhibiting regime substitution.
\end{proposition}

\begin{proof}
  Take rules, \(\tau=\RULEC\), \(\tau'=\RULES\),
  \(R=\{r_0,r_1,r_2,r_3\}\), and
  \[
    H
    =
    \bigl\langle
    (\textsf{refine-structure},\,r_0,\,r_1),\;
    (\textsf{revise-wording},\,r_0,\,r_2),\;
    (\textsf{amend-content},\,r_0,\,r_3)
    \bigr\rangle.
  \]
  Structural refinement preserves content and changes structure,
  so it is \PRS{} under \RULEC{} and \BRK{} under \RULES{}.
  Content-preserving rewording changes neither,
  so it is \PRS{} under both.
  A content amendment preserving section structure
  changes content and leaves structure fixed,
  so it is \BRK{} under \RULEC{} and \PRS{} under \RULES{}
  by Remark~\ref{oi.rem.norefine}.
  Therefore
  \[
    \pi_{\RULEC}
    =
    \bigl\{\{r_0,r_1,r_2\},\{r_3\}\bigr\},
    \qquad
    \pi_{\RULES}
    =
    \bigl\{\{r_0,r_2,r_3\},\{r_1\}\bigr\},
  \]
  and
  \[
    \Dsib
    =
    \mathrm{diag}(R)\cup\bigl\{(r_0,r_2),(r_2,r_0)\bigr\}.
  \]
  Let \(d\) be a registered mechanism with
  \[
    d(r_0)=d(r_2)=0,
    \qquad
    d(r_1)=d(r_3)=1,
  \]
  and a single use of referent-identifier assignment
  given by \(\tilde{u}(0)=\mathsf{a}\) and \(\tilde{u}(1)=\mathsf{b}\)
  for distinct referent identifiers \(\mathsf{a}\neq\mathsf{b}\).
  Surface control holds by construction,
  \(\IT_d=\tilde{u}\circ d\),
  and
  \[
    \pi_d
    =
    \bigl\{\{r_0,r_2\},\{r_1,r_3\}\bigr\}.
  \]
  Then
  \[
    \Dop
    \;=\;
    \Ker\;\cap\;\sim_{\RULEC}
    \;=\;
    \Dsib,
  \]
  so \(d\) exhibits sibling-aligned divergence.
  But \(\pi_d\neq\pi_{\RULES}\),
  since \(d\) separates \(r_0\) from \(r_3\)
  while \RULES{} merges them,
  and \(d\) merges \(r_1\) with \(r_3\)
  while \RULES{} separates them.
  The disagreement lies entirely on pairs that \RULEC{} separates,
  which \(\Dop\) and \(\Dsib\) do not see.
\end{proof}

\begin{proposition}[A Sibling-Separating Witness Does Not Establish Alignment]
  \label{oi.prop.witness-not-alignment}
  There is a carrier kind, a declared regime \(\tau\) with sibling \(\tau'\),
  a family-labeled history over the imported transformation basis,
  and a surface \(d\)
  such that a divergence witness separating under \(\tau'\) exists,
  and \(d\) exhibits sub-sibling divergence rather than sibling-aligned divergence.
\end{proposition}

\begin{proof}
  Take the carrier kind of rules,
  declared regime \(\tau=\RULEC\), sibling \(\tau'=\RULES\),
  and the record domain \(R=\{r_0,r_1,r_2\}\)
  with history
  \[
    H
    =
    \bigl\langle
    (\textsf{refine-structure},\,r_0,\,r_1),\;
    (\textsf{revise-wording},\,r_0,\,r_2)
    \bigr\rangle.
  \]
  Structural refinement preserves rule content and changes section structure,
  so it is \PRS{} under \RULEC{} and \BRK{} under \RULES{}.
  Content-preserving rewording changes neither content nor structure,
  so it is \PRS{} under both.
  Hence
  \[
    E_{\RULEC}=\{(r_0,r_1),(r_0,r_2)\},
    \qquad
    E_{\RULES}=\{(r_0,r_2)\},
  \]
  giving one declared block \(\{r_0,r_1,r_2\}\)
  and sibling blocks \(\{r_0,r_2\}\) and \(\{r_1\}\).

  Let \(d\) be a registered field holding a version counter
  that increments whenever the stored rule text changes,
  so \(d\) takes three distinct values on \(R\),
  and let its single use be referent-identifier assignment,
  returning a distinct referent identifier for each value.
  Surface control holds,
  and \(\Ker\) is the identity relation on \(R\).

  The pair \((r_0,r_1)\) satisfies \(r_0\sim_{\RULEC}r_1\)
  and \(\IT_d(r_0)\neq\IT_d(r_1)\),
  so it is a divergence witness,
  and \(r_0\not\sim_{\RULES}r_1\),
  so it separates under the sibling.

  The pair \((r_0,r_2)\) satisfies \(r_0\sim_{\RULEC}r_2\)
  and \(r_0\sim_{\RULES}r_2\)
  while \(\IT_d(r_0)\neq\IT_d(r_2)\),
  so \((r_0,r_2)\in\Dsib\setminus\Dop\).
  Hence \(\Dop\subsetneq\Dsib\),
  which is sub-sibling divergence.
\end{proof}

The construction uses only a version counter
that increments on any change to stored text.
On the examined records it splits declared classes
strictly more finely than either imported basis.
The pair \((r_0,r_1)\) is a valid divergence witness,
but it does not establish the classification it appears to invite.
Section~\ref{oi.sec.worked} shows why the classification matters:
the repair suggested by sibling alignment
would leave this fault in place.

\begin{remark}[Non-vacuity of the Classification]
  \label{oi.rem.nonvacuity}
  Cases (a) and (b) of Definition~\ref{oi.def.classification}
  are realized by Example~\ref{oi.ex.loc-obj-fork}
  and Proposition~\ref{oi.prop.witness-not-alignment}.
  Cases (c) and (d) also occur under the imported interface,
  so the four-way classification is not partly empty.

  \emph{Super-sibling (c).}
  Take \(\tau=\RULEC\), \(\tau'=\RULES\),
  \(R=\{r_0,r_1,r_2\}\), and
  \[
    H
    =
    \bigl\langle
    (\textsf{refine-structure},\,r_0,\,r_1),\;
    (\textsf{refine-structure},\,r_0,\,r_2)
    \bigr\rangle.
  \]
  Both edges are \PRS{} under \RULEC{} and \BRK{} under \RULES{},
  so
  \[
    \pi_{\RULEC}=\bigl\{\{r_0,r_1,r_2\}\bigr\},
    \qquad
    \pi_{\RULES}=\bigl\{\{r_0\},\{r_1\},\{r_2\}\bigr\},
  \]
  giving \(\Dsib=\mathrm{diag}(R)\) on the declared class.
  Let \(d\) be a registered field with
  \(d(r_0)=d(r_1)=0\) and \(d(r_2)=1\),
  and a single use of referent-identifier assignment
  given by \(\tilde{u}(0)=\mathsf{a}\) and \(\tilde{u}(1)=\mathsf{b}\)
  for distinct referent identifiers \(\mathsf{a}\neq\mathsf{b}\).
  Then \(\Ker\) has blocks \(\{r_0,r_1\}\) and \(\{r_2\}\), so
  \[
    \Dop
    =
    \mathrm{diag}(R)\cup\bigl\{(r_0,r_1),(r_1,r_0)\bigr\}
    \qquad\text{and}\qquad
    \Dsib\subsetneq\Dop\subsetneq\;\sim_{\RULEC}.
  \]
  The surface makes one of the two structural distinctions
  the sibling makes and omits the other.

  \emph{Sibling-incomparable (d).}
  Take the same \(\tau,\tau'\),
  with \(R=\{r_0,r_1,r_2,r_3\}\) and
  \[
    H
    =
    \bigl\langle
    (\textsf{revise-wording},\,r_0,\,r_1),\;
    (\textsf{refine-structure},\,r_0,\,r_2),\;
    (\textsf{revise-wording},\,r_2,\,r_3)
    \bigr\rangle.
  \]
  All three edges are \PRS{} under \RULEC{},
  so \(\pi_{\RULEC}=\bigl\{\{r_0,r_1,r_2,r_3\}\bigr\}\).
  Under \RULES{} the two rewordings are \PRS{}
  and the refinement is \BRK{},
  so \(\pi_{\RULES}=\bigl\{\{r_0,r_1\},\{r_2,r_3\}\bigr\}\),
  and \(\Dsib\) merges \((r_0,r_1)\) and \((r_2,r_3)\).
  Let \(d\) be a registered field with
  \(d(r_0)=d(r_2)=0\) and \(d(r_1)=d(r_3)=1\),
  and a single use of referent-identifier assignment
  given by \(\tilde{u}(0)=\mathsf{a}\) and \(\tilde{u}(1)=\mathsf{b}\)
  for distinct referent identifiers \(\mathsf{a}\neq\mathsf{b}\).
  So \(\Ker\) has blocks \(\{r_0,r_2\}\) and \(\{r_1,r_3\}\).
  Then
  \[
    (r_0,r_1)\in\Dsib\setminus\Dop
    \qquad\text{and}\qquad
    (r_0,r_2)\in\Dop\setminus\Dsib,
  \]
  so \(\Dop\) and \(\Dsib\) are incomparable.
  In both constructions surface control holds through the single use,
  and \(d\) splits the declared class,
  so the audit returns \FAIL{}.
\end{remark}

Propositions~\ref{oi.prop.aligned-not-substitution}
and~\ref{oi.prop.witness-not-alignment}
bound the diagnostic strength of the audit from both sides.
A witness does not establish a classification, and
a classification does not establish a basis.
The audit establishes that the declaration and the implementation
disagree and, when the sibling comparison is informative,
where their disagreement lies in the lattice.

\section{Decidability, Verdict, and Non-monotonicity}
\label{oi.sec.procedure}

\begin{assumption}[Finite Audit Artifact]
  \label{oi.asm.finite}
  For each examined surface \(d\):
  \begin{enumerate}[label=(\alph*)]
    \item the preconditions of Definition~\ref{oi.def.precond} are discharged;
    \item \(R\) and the relevant transformation history \(H\) are finite;
    \item every history endpoint resolves to a record in \(R\);
    \item each family \(f_i\) is known and \(\Class_{\rho}(f_i)\) is decidable
          for each regime assigned to the carrier kind;
    \item \(d\in A_B\), so \(d\in\mathcal{D}_B\)
          and \(d(r)\) and \(\IT_d(r)\) are computable for every \(r\in R\);
    \item surface control is established for \(d\) over \(\mathcal{D}_B\).
  \end{enumerate}
  Store-level statements range over the examined family \(A_B\).
\end{assumption}

\begin{theorem}[Finite Decidability]
  \label{oi.thm.decidable}
  Under Assumption~\ref{oi.asm.finite},
  the declared partition \(\pi_{\tau}\),
  the operational partition \(\pi_d\),
  faithfulness of \(d\) to \(\tau\) on \(R\),
  existence of a divergence witness,
  and the divergence classification of Definition~\ref{oi.def.classification}
  are all computable.

  Let
  \[
    n=|R|,
    \qquad
    h=|H|,
    \qquad
    m=|U_d|.
  \]
  Constructing the declared and sibling partitions requires at most \(2h\)
  calls to the imported classification functions
  and \(O((n+h)\alpha(n))\) union-find time,
  where \(\alpha\) is the inverse-Ackermann function.
  Constructing the operational partition requires \(mn\)
  identified-use evaluations and grouping the resulting signatures.
  Once the partitions are built,
  faithfulness, witness extraction, and divergence classification
  take at most \(O(n^2)\) time by pair enumeration
  and expected \(O(n)\) time using block labels and hashing.
\end{theorem}

\begin{proof}
  The checker computes \(E_{\tau}\) and, when \(\Sib(\tau)\neq\varnothing\),
  \(E_{\tau'}\), from the finite history,
  using decidability of the classification functions.
  It forms \(\sim_{\tau}\) and \(\sim_{\tau'}\)
  as equivalence closures by union-find over the finite edge sets,
  giving \(\pi_{\tau}\) and the sibling partition
  in \(O((n+h)\alpha(n))\) union-find time.

  It computes \(\IT_d(r)\) for every \(r\in R\)
  and groups records by signature, giving \(\pi_d\).
  This requires \(mn\) identified-use evaluations.
  Apart from the costs of evaluating and representing individual outcomes,
  constructing and hashing the signatures takes expected \(O(mn)\) time.

  Faithfulness is the containment
  \(\sim_{\tau}\subseteq\;\Ker\).
  It is decidable by enumerating the pairs
  in each block of \(\pi_{\tau}\)
  and comparing their signatures.
  Any pair at which the containment fails
  is a divergence witness.

  Suppose faithfulness fails.
  When \(\Sib(\tau)=\varnothing\),
  the divergence is computably classified as unpositioned.
  When a sibling exists, the checker computes \(\Dsib\);
  if \(\Dsib=\;\sim_{\tau}\),
  the divergence is again classified as unpositioned.
  Otherwise, the checker computes $\Dop$
  and performs the containment tests
  \[
    \Dop\subseteq\Dsib
    \qquad\text{and}\qquad
    \Dsib\subseteq\Dop.
  \]
  The four possible outcomes are exactly
  the four sibling-relative cases of
  Definition~\ref{oi.def.classification}.

  Pair enumeration gives a deterministic \(O(n^2)\)
  upper bound for this comparison stage.
  For the near-linear implementation,
  assign each record its block labels in
  \(\pi_{\tau}\), \(\pi_d\), and, when present, \(\pi_{\tau'}\).
  Faithfulness holds exactly when every declared block
  has a single operational label;
  retaining one representative record for each declared block
  produces a witness when a second label is encountered.
  The containment \(\Dop\subseteq\Dsib\) holds exactly when
  every block identified by its declared and operational labels
  has a single sibling label.
  The reverse containment holds exactly when
  every block identified by its declared and sibling labels
  has a single operational label.
  Hash tables perform these scans in expected \(O(n)\) time.
\end{proof}

\subsection{Three-Valued Verdict}

\begin{definition}[Audit Verdict]
  \label{oi.def.verdict}
  The audit of a surface \(d\) against declared regime \(\tau\) on \(R\) returns:
  \begin{itemize}
    \item \PASS{} when the preconditions and surface control are discharged
          and \(d\) is faithful to \(\tau\) on \(R\);
    \item \FAIL{} when they are discharged and a divergence witness exists,
          reported with the classification of
          Definition~\ref{oi.def.classification}
          and with a witness pair;
    \item \INDET{} when a precondition is not discharged,
          surface control cannot be established over \(\mathcal{D}_B\),
          a family label is unavailable or unverified,
          or an identity-relevant outcome cannot be computed
          for some record in \(R\).
  \end{itemize}
\end{definition}

\begin{remark}[Indeterminate Is Not Conformance]
  \label{oi.rem.indet}
  An \INDET{} verdict reports that the audit could not decide.
  It has the standing of an unexamined surface
  and does not support a claim that the surface is faithful.
  Privacy, security, trade-secret, and access constraints
  frequently produce \INDET{} rather than \PASS{},
  and reporting the two under one heading
  would let an unexaminable system present as a conformant one.
\end{remark}

\subsection{Passing Is Not Monotone}

The declared partition is built by transitive closure.
That fact has a consequence for what a passing verdict can be indexed to.

\begin{proposition}[Non-monotonicity of \PASS{} under History Extension]
  \label{oi.prop.nonmonotone}
  There exist a fixed record domain \(R\),
  histories \(H\preceq H^{\ast}\) over \(R\),
  and a surface \(d\)
  such that \(d\) is faithful to \(\tau\) on \((R,H)\)
  and a divergence witness for \((d,\tau)\) exists on \((R,H^{\ast})\).
  The record domain, the registry, the surface,
  the records of the witness,
  and their identity treatment are all unchanged.
\end{proposition}

\begin{proof}
  Let \(R=\{r,s,t\}\) and let \(H=\langle\rangle\) be the empty history,
  so \(E_{\tau}=\varnothing\),
  the relation \(\sim_{\tau}\) is the identity relation on \(R\),
  and \(\pi_{\tau}\) consists of three singleton blocks.
  Let \(d\) be a registered mechanism taking three distinct values on \(R\),
  with a referent-identifier use returning a distinct identifier for each,
  so \(\IT_d\) separates all three records.
  No pair \((x,y)\) with \(x\neq y\) satisfies \(x\sim_{\tau}y\),
  so no divergence witness exists
  and \(d\) is faithful on \((R,H)\).

  Extend the history alone to
  \[
    H^{\ast}
    =
    \bigl\langle
    (f,\,r,\,t),\;
    (g,\,t,\,s)
    \bigr\rangle,
  \]
  where \(f\) and \(g\) are families non-breaking under \(\tau\).
  The record domain is still \(R\),
  and both endpoints of both transformations lie in it.
  Now \((r,t)\) and \((t,s)\) lie in \(E_{\tau}\),
  and the equivalence closure gives \(r\sim_{\tau}s\),
  so \(\pi_{\tau}\) has collapsed to the single block \(\{r,s,t\}\).
  By surface control \(\IT_d\) is a function of \(d\),
  so the signatures are unchanged
  and \(\IT_d(r)\neq\IT_d(s)\) still holds.
  The pair \((r,s)\) is now a divergence witness on \((R,H^{\ast})\).
\end{proof}

\begin{corollary}[Indexing of a Passing Verdict]
  \label{oi.cor.indexing}
  For a fixed examined surface and identified-use sequence,
  a \PASS{} verdict is indexed to
  \[
    (\mathfrak{R},R,H)
  \]
  and not to \(R\) alone.
  The audit must be rerun when the transformation history is extended
  or the imported referential-regime interface is revised,
  even when the record domain, the registry, and the surface are unchanged.

  This index records the history and interface dependence only.
  A pass verdict is additionally relative to
  the disclosed registry \(D_B\),
  the evaluated surface family \(A_B\),
  and the identified uses \(U_d\) of each surface;
  Section~\ref{oi.sec.disclosure} states those dependences
  and the witnesses that refute them.
\end{corollary}

Non-monotonicity distinguishes this audit from a test
whose passing result is stable under the arrival of new data.
With \(R\) and the operational signatures fixed,
the declared partition can only coarsen
as history accumulates.
That coarsening can create divergence witnesses
but cannot remove an existing one.
An operator who obtains a passing verdict
holds a statement about a history,
not a certificate about a store.

\section{Worked Legal-Alignment Example}
\label{oi.sec.worked}

Consider a legal-alignment system
recording a reporting rule used to evaluate
a regulatory filing submitted by an AI-enabled agent.
Advanced agents may act within legal and institutional processes,
and Kolt describes a trajectory in which they function as
subjects, consumers, producers, and enforcers of law~\citep{kolt2026superintelligence}.
Empirical survey of deployed agentic systems reports
uneven documentation and limited disclosure across the ecosystem~\citep{staufer2026agentindex},
which complicates later inspection
of which rules, filings, and authorities an accountability claim concerns.

The legal conclusion may remain contested.
The identity of the rule must remain stable enough
for reviewers to examine the same authority
while disputing whether the filing complied with it.

\subsection{Records and Declared Partition}

Let \(r_0\) be the original rule record.
Let \(r_1\) be a structural refinement:
it adds subsections and renumbers sections
while preserving the rule content.
Let \(r_2\) be an editorial revision of \(r_0\)
that rewords a sentence for clarity
while preserving both content and section structure.

The system declares content-fixed rule identity, \(\RULEC\).
Its imported sibling is structure-fixed rule identity, \(\RULES\).
By the classification used in
Proposition~\ref{oi.prop.witness-not-alignment},
the declared partition is the single block
\[
  \pi_{\RULEC}
  =
  \bigl\{\{r_0,r_1,r_2\}\bigr\},
\]
and the sibling partition is
\[
  \pi_{\RULES}
  =
  \bigl\{\{r_0,r_2\},\{r_1\}\bigr\}.
\]
All three records refer to one rule under the declaration.
The sibling basis would separate the structurally refined expression.

\subsection{The Surface and Its Operational Partition}

The registry \(D_B\) contains the field
\[
  d=\texttt{ruleTextVersion},
\]
which increments whenever the stored rule text changes.
All three records carry distinct values.

A value difference alone is harmless.
The field may identify a displayed expression,
support provenance,
or route a reviewer to the relevant text.
It becomes a surface only when identity-relevant uses factor through its value.

Suppose the application binds the referent identifier of a rule
to the value of \texttt{ruleTextVersion},
so that when the field changes the application
issues a new rule identifier,
refuses co-reference with the previous expression,
and resolves later citations to the new rule.
Those uses are functions of the field value,
so surface control holds.
The operational identity partition is
\[
  \pi_d
  =
  \bigl\{\{r_0\},\{r_1\},\{r_2\}\bigr\}.
\]

\subsection{Comparison}

The declared partition does not refine the operational one:
the block \(\{r_0,r_1,r_2\}\) is split.
The surface is not faithful, and the audit returns \FAIL{}.
The pair \((r_0,r_1)\) is a divergence witness,
and it separates under the sibling.

Restricting to the declared class gives
\[
  \Dop
  =
  \bigl\{(r_0,r_0),(r_1,r_1),(r_2,r_2)\bigr\},
  \qquad
  \Dsib
  =
  \Dop\cup\bigl\{(r_0,r_2),(r_2,r_0)\bigr\},
\]
so \(\Dop\subsetneq\Dsib\).
The classification is sub-sibling divergence.

The operational and sibling partitions differ:
\(\pi_d=\bigl\{\{r_0\},\{r_1\},\{r_2\}\bigr\}\)
separates \(r_0\) from \(r_2\),
which \(\pi_{\RULES}=\bigl\{\{r_0,r_2\},\{r_1\}\bigr\}\) merges.

Therefore \(\approx_d\neq{\sim_{\RULES}}\) on \(R\),
so the surface does not exhibit regime substitution
(Definition~\ref{oi.def.regime-substitution}),
and the audit reports no hidden identity regime.

\subsection{Why the Classification Changes the Repair}

Reported as sibling alignment, the finding invites one repair:
declare structure-fixed rule identity and document the basis.
Applied here, that repair leaves the fault in place.
The application would continue to break co-reference on \(r_2\),
an editorial revision that structure-fixed identity preserves.
The operator would have declared \RULES{} and still not be running it.

Reported as sub-sibling divergence, the finding says what is true.
On the examined records the operational partition is total,
so any change to the stored text yields a distinct rule referent.
The split is strictly finer than either imported basis,
and the imported inventory contains no basis that produces it.
The audit does not name the basis the application carries,
and the natural reading of \(\pi_d\) on \(R\)
is text-fixed rule identity,
which the operator can confirm from the field's update rule.

Three repairs follow, and the audit chooses among none of them.
The first preserves the declared regime:
the version field remains available for display and provenance,
the referent identifier is unbound from it,
and every content-preserving expression resolves to the same rule.
The surface values continue to differ,
the operational partition coarsens to the declared one,
and the surface becomes faithful.
The second declares structure-fixed identity
and repairs the implementation to match it,
which requires the field to stop firing on editorial revision.
The third adopts text-fixed identity as the intended basis,
which requires an argument that it is stability-critical
and a statement of its transformation classification,
and which places the declared basis outside the imported inventory.

The audit establishes that the declaration
and the examined implementation behavior diverge.
It also establishes which distinctions
the implementation carries on the examined records
and where those distinctions lie
relative to the declared and sibling partitions.
It does not name the identity basis responsible for them.

\section{Discovery and the Disclosure Boundary}
\label{oi.sec.disclosure}

The procedure evaluates the surfaces in the registry.
It does not discover every mechanism inside arbitrary source code, workflow configuration,
remote services, operational practice, or undocumented organizational convention.

Candidate mechanisms may be found through
source-code analysis,
schema and configuration inspection,
workflow-rule extraction,
runtime tracing,
differential testing,
policy-engine inspection,
legal discovery,
or post-incident review.
Those methods answer
\emph{where might identity treatment be controlled?}
The procedure of Section~\ref{oi.sec.procedure} answers
\emph{which rule of sameness does this mechanism carry,
  and is it the declared one?}

\begin{definition}[Registry-Completeness Claim]
  \label{oi.def.completeness}
  A \emph{registry-completeness claim} for \((B,D_B)\)
  asserts that every mechanism within \(B\)
  that affects identity-relevant treatment
  is represented in \(D_B\).
\end{definition}

The procedure can evaluate every surface in the examined family \(A_B\).
It cannot derive the truth of the completeness claim from \(D_B\) itself.

\begin{definition}[Omitted-Mechanism Witness]
  \label{oi.def.omitted}
  An \emph{omitted-mechanism witness} for \((B,D_B)\) consists of:
  \begin{enumerate}[label=(\alph*)]
    \item a mechanism \(d^{\ast}\) operating within \(B\) with \(d^{\ast}\notin D_B\);
    \item the declared regime \(\tau\) for the carrier kind of its record domain;
    \item records \(r,s\) within the examined system or preserved audit history;
    \item evidence that the identity-relevant uses of \(d^{\ast}\)
          factor through its value;
    \item evidence that \(r\sim_{\tau}s\); and
    \item evidence that \(\IT_{d^{\ast}}(r)\neq\IT_{d^{\ast}}(s)\).
  \end{enumerate}
\end{definition}

The omitted-mechanism witness requires divergence and does not require a sibling comparison.
An undisclosed mechanism that splits a declared co-reference class
is an identity-relevant mechanism absent from the disclosure,
whatever basis it turns out to carry.

\begin{proposition}[Later Witness Refutes Completeness]
  \label{oi.prop.refutes-completeness}
  An omitted-mechanism witness within \(B\)
  refutes the registry-completeness claim for \((B,D_B)\).
\end{proposition}

\begin{proof}
  The witness exhibits a mechanism within \(B\) and absent from \(D_B\).
  Clauses (d) through (f) establish that it controls a difference
  in identity-relevant treatment between records the declaration merges,
  so it affects identity-relevant treatment.
  The claim that \(D_B\) contains every such mechanism within \(B\) is false.
\end{proof}

\subsection{Family Completeness}

Registry completeness bounds which implementation artifacts the audit has seen.
It does not bound which \emph{surfaces} over those artifacts the audit evaluated.

The gap is created by the closure \(\mathcal{D}_B\).
An identity-relevant use may factor through the joint mechanism \(d_{\{1,2\}}\)
and through neither \(d_1\) nor \(d_2\) alone.
Suppose \(D_B=\{d_1,d_2\}\), suppose the registry-completeness claim is true,
and suppose \(A_B\) contains only the two singleton surfaces.
The use-completeness claim for each singleton can be true,
because the jointly controlled use is not a use of either singleton
in the factorization sense of Definition~\ref{oi.def.surface}(c).
Each examined surface can pass.
The joint surface can nevertheless split a declared co-reference class,
and \(\pi_B\) does not see it.
A simpler instance is a registered artifact
the auditor never admits into \(A_B\) at all.

\begin{definition}[Family-Completeness Claim]
  \label{oi.def.familycomplete}
  A \emph{family-completeness claim} for \((B,D_B,A_B)\)
  asserts that, for every identity-relevant use within \(B\)
  whose controlling mechanisms are represented in \(D_B\),
  \(A_B\) contains at least one audit surface
  through whose mechanism that use factors.
\end{definition}

\begin{definition}[Omitted-Surface Witness]
  \label{oi.def.omittedsurface}
  An \emph{omitted-surface witness}
  for a triple \((B,D_B,A_B)\)
  consists of:
  \begin{enumerate}[label=(\alph*)]
    \item an identity-relevant use \(u^{\ast}\) operating within \(B\);
    \item a nonempty set \(S\subseteq D_B\),
          the associated joint mechanism
          \(d_S\in\mathcal{D}_B\),
          and a map \(\tilde{u}^{\ast}\) such that
          \[
            u^{\ast}
            =
            \tilde{u}^{\ast}\circ d_S
            \quad\text{on }R;
          \]
    \item for every examined surface
          \((d,U_d)\in A_B\),
          there is no map \(v\) such that
          \[
            u^{\ast}=v\circ d
            \quad\text{on }R;
          \]
    \item records \(r,s\in R\) with
          \(r\sim_{\tau}s\); and
    \item evidence that
          \(u^{\ast}(r)\neq u^{\ast}(s)\).
  \end{enumerate}
\end{definition}

\begin{proposition}[Omitted Surface Refutes Family Completeness]
  \label{oi.prop.refutes-family}
  An omitted-surface witness
  refutes the family-completeness claim
  for its triple \((B,D_B,A_B)\).
  Suppose the audit returned \PASS{}
  for every surface in \(A_B\),
  and let
  \[
    A_B^{\ast}
    =
    A_B\cup\bigl\{(d_S,\langle u^{\ast}\rangle)\bigr\}.
  \]
  Then the audit returns \FAIL{} for \(d_S\),
  and by Proposition~\ref{oi.prop.compose}
  the store verdict over \(A_B^{\ast}\) is \FAIL{}.
\end{proposition}

\begin{proof}
  Clause (b) establishes that the controlling artifacts
  of \(u^{\ast}\) are represented in \(D_B\).
  Clause (c) establishes that
  \(u^{\ast}\) factors through no examined surface mechanism.
  The family-completeness claim is therefore false.

  Clause (b) also establishes surface control for
  \((d_S,\langle u^{\ast}\rangle)\),
  so it is an admissible audit surface.
  Clauses (d) and (e) make \((r,s)\)
  a divergence witness for that surface.
\end{proof}

The omitted-surface witness is not an omitted-mechanism witness.
Its constituent artifacts are already registered,
so the registry-completeness claim survives it intact.
What it refutes is the coverage of the evaluation,
and the repair is to evaluate the joint surface,
not to disclose a new artifact.

\subsection{Use Completeness}

The two prior boundaries bound which mechanisms the audit has seen
and which surfaces over them it evaluated.
A third bounds which \emph{uses} of an evaluated surface it has seen.

The operational partition \(\pi_d\) is computed from \(U_d\),
the identified identity-relevant uses of \(d\).
Nothing in the computation requires \(U_d\)
to contain every identity-relevant use of \(d\) within \(B\).
An operator who discloses a harmless use of a field
and omits the use that binds it to referent-identifier assignment
supplies a coarser \(\pi_d\),
and the audit may return \PASS{} on a surface that in fact diverges.

\begin{definition}[Use-Completeness Claim]
  \label{oi.def.usecomplete}
  A \emph{use-completeness claim} for \((d,B,U_d)\)
  asserts that every identity-relevant use of \(d\) within \(B\),
  in the sense of Definition~\ref{oi.def.outcomes},
  appears in \(U_d\).
\end{definition}

\begin{definition}[Omitted-Use Witness]
  \label{oi.def.omitteduse}
  An \emph{omitted-use witness} for \((d,B,U_d)\) consists of:
  \begin{enumerate}[label=(\alph*)]
    \item a use \(u^{\ast}\) operating within \(B\)
          with \(u^{\ast}\notin U_d\),
          of an outcome type in Definition~\ref{oi.def.outcomes};
    \item \emph{surface control:} a map \(\tilde{u}^{\ast}\) with
          \[
            u^{\ast}=\tilde{u}^{\ast}\circ d
            \quad\text{on }R,
          \]
          so the omitted use is a use \emph{of \(d\)} in the sense of
          Definition~\ref{oi.def.surface}(c);
    \item records \(r,s\in R\) with \(r\sim_{\tau}s\); and
    \item evidence that \(u^{\ast}(r)\neq u^{\ast}(s)\).
  \end{enumerate}
\end{definition}

Clause (b) is what keeps the expanded surface admissible.
An identity-relevant outcome that does not factor through \(d\)
is not an omitted use of \(d\).
It is jointly controlled,
and Remark~\ref{oi.rem.granularity} routes it
to a joint surface in \(\mathcal{D}_B\).
When the constituent artifacts are already registered,
that outcome is an omitted-\emph{surface} finding
under Definition~\ref{oi.def.omittedsurface},
a gap in the coverage of the evaluation
rather than a gap in the disclosure of artifacts.
It becomes an omitted-mechanism finding
only when a controlling artifact is absent from \(D_B\).

\begin{proposition}[Omitted Use Refutes Completeness and Overturns a Pass]
  \label{oi.prop.refutes-use}
  An omitted-use witness for \((d,B,U_d)\)
  refutes the use-completeness claim for \((d,B,U_d)\).
  If the audit returned \PASS{} for \(d\) on \((R,H)\) under \(U_d\),
  then the audit over
  \(U_d^{\ast}=U_d\frown\langle u^{\ast}\rangle\)
  returns \FAIL{},
  and \((r,s)\) is a divergence witness for it.
\end{proposition}

\begin{proof}
  The witness exhibits an identity-relevant use of \(d\) within \(B\)
  that is absent from \(U_d\),
  which refutes the claim.
  By clause (b) the use factors through \(d\),
  so \(\bigl(d,U_d^{\ast}\bigr)\) satisfies
  Definition~\ref{oi.def.surface}(c)
  and remains an audit surface.
  Under \(U_d^{\ast}\) the signature of \(r\) extends by \(u^{\ast}(r)\)
  and that of \(s\) by \(u^{\ast}(s)\),
  which differ,
  so \(\IT_d(r)\neq\IT_d(s)\) while \(r\sim_{\tau}s\).
\end{proof}

The three completeness claims are independent, and the audit is relative to all of them.
Registry completeness can hold while family completeness fails,
and the mechanism responsible is then one the operator disclosed
and the auditor never evaluated.
Registry and family completeness can both hold while use completeness fails,
and the surface responsible is then one the audit examined and cleared.

\begin{remark}[What Licenses the System-Level Reading]
  \label{oi.rem.systemreading}
  Proposition~\ref{oi.prop.compose} combines the surface verdicts
  into the store operational partition \(\pi_B\).
  That composition is sound over the examined family \(A_B\)
  and over each \(U_d\).
  Reading \(\pi_B\) as the rule of sameness the system implements,
  rather than as the rule of sameness
  the evaluated surfaces and their identified uses implement,
  requires three claims:
  \begin{enumerate}[label=(\arabic*)]
    \item registry completeness,
          that \(D_B\) contains every controlling artifact within \(B\);
    \item family completeness,
          that \(A_B\) contains a surface through whose mechanism
          every identity-relevant use factors
          when its controlling artifacts lie in \(D_B\); and
    \item use completeness for each \(d\in A_B\),
          that \(U_d\) contains every identity-relevant use factoring through \(d\).
  \end{enumerate}
  None is established by the procedure.
  Each is refutable by a finite witness,
  and none is provable from the artifact that asserts it.
\end{remark}

\subsection{Scopes of a Finding}

Audit findings therefore have five distinct scopes.
A divergence witness for a surface in \(A_B\)
identifies a fault in the examined implementation,
classified by Definition~\ref{oi.def.classification}.
An omitted-mechanism witness
identifies an incomplete account of the implementation boundary.
An omitted-surface witness
identifies an incomplete evaluation of artifacts already disclosed,
and by Proposition~\ref{oi.prop.refutes-family}
it can overturn a store-level passing verdict
without any artifact having been withheld.
An omitted-use witness
identifies an incomplete account of a surface already examined,
and by Proposition~\ref{oi.prop.refutes-use}
it can overturn a passing verdict on that surface.
A passing verdict establishes only that
no disclosed use of any examined surface
splits a declared co-reference class
on the examined records under the examined history,
and Corollary~\ref{oi.cor.indexing} bounds even that.

\section{Related Work}
\label{oi.sec.related}

\subsection{Entity Resolution and Record Linkage}

Record linkage asks whether two records describe the same real-world entity,
and the probabilistic framework of Fellegi and Sunter
formalized the decision as a likelihood-ratio test
over agreement patterns between record fields~\citep{fellegisunter1969}.
The field has developed extensive machinery for
similarity measures, classification, and evaluation
~\citep{elmagarmid2007duplicate,christen2012datamatching},
for blocking and filtering schemes
that partition records into candidate groups
before pairwise comparison~\citep{papadakis2020blocking},
and for the scaling and quality challenges
of resolution at practical volume~\citep{getoor2012er}.

The resemblance to the present result is real and worth stating plainly.
A blocking key is a mechanism whose value determines
whether two records will be compared,
and it therefore induces a partition of the record domain.
A blocking key that separates true matches
is a mechanism carrying an identity distinction its designer did not intend,
which is structurally the fault this paper isolates.

Three differences separate the settings.

The comparison in entity resolution is against ground truth.
Records are labeled as matches or non-matches by an external standard,
and the matching function is evaluated for error against those labels.
The present result has no ground truth and asks for none.
It compares two artifacts the system itself supplies:
what it declares and what it does.
A divergence witness is a disagreement internal to the system,
not an error against an outside label,
and it is available to an auditor
who has no view on which records truly co-refer.

Co-reference is derived differently.
Entity resolution derives it from attribute similarity
across a pair of records considered on their own.
The declared partition here is derived from a transformation history:
records co-refer when a chain of transformations connects them,
each of whose families is non-breaking under the declared basis.
Two records with identical attributes may fail to co-refer
when no history connects them,
and two records with quite different attributes may co-refer
when a chain of content-preserving transformations does.

The object under evaluation differs.
Entity resolution evaluates a matching function.
This paper evaluates a partition induced by an implementation mechanism
against a partition induced by a declared identity regime,
and diagnoses their divergence against an inventory of identity bases.
Where entity resolution asks how well a matcher performs,
this paper asks which rule of sameness a system is running.

\subsection{Provenance and Data Lineage}

PROV-DM provides a general model of entities,
activities, agents, derivations, generations, uses, associations,
and attributions~\citep{provdm2013}.
It supplies the relations needed to identify transformation occurrences and their evidentiary basis.

A derivation relation does not determine
whether source and result represent the same referent.
A complete provenance history may record
that one artifact was derived from another
without stating whether the transformation preserved
or broke identity under the system's chosen basis.
The present result uses family-labeled histories
to derive the declared partition,
and then compares that partition with behavior.
The question is not only what happened,
but which rule of sameness the system operationally applied.

\subsection{Identity Conditions and Ontology Engineering}

Foundational ontologies provide categories
for entities, processes, qualities, roles, and related distinctions~\citep{arp2015building,gangemi2002dolce,masolo2003wonderweb}.
Formal ontology and OntoClean examine ontological commitment,
identity conditions, rigidity, dependence, and the quality of taxonomic choices
~\citep{guarino1998formalontology,guarino1999roles,guarinowelty2002,guizzardi2005}.

Those approaches support careful selection of what kinds of things a system represents
and which conditions govern their identity.
The problem here begins after that selection.
The declared identity condition may be sound while workflow or application logic implements another,
and the operational identity partition is the object that makes the second condition visible.

\subsection{AI-System Identity and Governance}

Ferrario develops a trustworthiness-based metaphysics of artificial intelligence systems
in which AI-system kinds are fixed by techno-function
and synchronic and diachronic identity are governed
by trustworthiness profiles and compatible trustworthiness levels
~\citep{ferrario2025metaphysics}.
Applied to high-risk systems under the European AI Act,
the account makes explicit that lifecycle governance
presupposes judgments about whether an updated or separately deployed system
remains the same system for regulatory purposes.
Ferrario argues that the Act supplies no internal, auditable criterion for synchronic identity
and proposes a correspondence map and minimal decision flow
for applying the function-plus-trustworthiness criterion
in audit and dispute settings~\citep{ferrario2026aia}.

A subsequent category-theoretic account fixes an AI-system type datum consisting of
a techno-function, trustworthiness profile, and trustworthiness-level function.
It distinguishes equality of trustworthiness level,
directed trustworthiness-preserving reachability,
mutual reachability as state isomorphism,
and natural isomorphism of realized system histories
~\citep{ferrario2026category}.
That account provides a formal hierarchy of weaker and stronger identity criteria,
while leaving their detailed recognition in operational governance and MLOps settings
to further methodology.

These works make one substantive AI identity criterion explicit, inspectable, and formally structured.
The present result addresses a different audit object.
It neither selects the correct substantive identity basis
nor determines whether trustworthiness is
the appropriate basis for AI-system identity.
It takes the effective declared regime as an input
and asks whether disclosed implementation surfaces
induce identity-relevant treatment faithful to that declaration.
Its output is therefore not an AI identity decision flow,
but a comparison of declared and operational partitions
and, on failure, a finite pair witnessing a distinction made by the implementation
that the declaration does not make.

\subsection{Classification and Institutional Infrastructure}

Classification systems shape
the practices and institutions that use them~\citep{bowkerstar1999}.
Longino emphasizes the social and institutional conditions
under which claims become open to criticism and revision~\citep{longino1990}.

A published classification may differ from the classification implemented by the system.
The operational identity partition converts that institutional concern
into a finite computed object for one class of stability-critical distinctions.

\subsection{Versioning, Digital Twins, and Persistent Objects}

Version-control systems,
temporal databases,
and digital-twin architectures distinguish
states, revisions, derivations, and evolving objects.
Digital-twin work emphasizes continuity
between physical or operational systems
and their changing digital representations
~\citep{grievesvickers2017,voas2025digitaltwin}.

These systems may preserve detailed histories while leaving the identity criterion implicit.
Version succession does not by itself determine referent identity.
Section~\ref{oi.sec.worked} shows a version counter
carrying an identity criterion finer than any basis
the system's own inventory supplies,
which is a case the versioning literature has the machinery to express
and no standing reason to look for.

\subsection{Contradiction Detection}

Recent accountability protocols for adversarial supply chains
show the value of signed event claims,
append-only causal histories,
and contradiction proof objects.
ECO/CPO-DAG~\citep{cochinescu2026ecocpo},
for example,
treats contradiction detection as a supplemental validation layer
rather than a truth-establishing consensus mechanism,
and compiles a finite, self-verifying object
that binds two signed claims to the violated rule.

The present result shares that commitment to inspectable finite witnesses rather than scores.
The approaches differ along an axis the protocol itself marks:
\emph{a party that never contradicts itself
  is invisible to contradiction detection.}

A system that consistently applies an undeclared identity rule
for what counts as the same lot, the same consignment,
the same rule, or the same regulated entity
produces no contradiction and emits no contradiction proof object.
Contradiction detection finds inconsistent claims within a declared rule system.
The divergence witness finds a rule of sameness the system never declared.

\subsection{Provenance and Accountability for AI Agents}

Recent work argues that accountability for AI agents
requires structural provenance rather than isolated component logs.
Hu et al.\ position explicit provenance across the agentic lifecycle
as necessary for assigning responsibility
when harm emerges from compositions
that no single party designed~\citep{hu2026provenance}.
Ojewale, Suresh, and Venkatasubramanian
propose audit trails as a tamper-evident, system-agnostic ledger
linking technical provenance to governance records~\citep{ojewale2026audit}.
PROV-AGENT extends W3C PROV to capture prompts, responses, decisions,
and agent interactions across end-to-end workflows~\citep{souza2025provagent}.
Work on auditable agents treats auditability as the system property
that makes accountability possible~\citep{nian2026auditable}.

These approaches expose the provenance an audit needs.
The operational identity partition supplies a complementary question:
whether the co-reference and persistence behavior the system exhibits
follows the identity basis it declares.

Otsuka, Toyoda, and Leung define AI identity through correspondence
between what an agent is declared to be
and what it is observed to do~\citep{otsuka2026identity}.
Their declaration-observation gap is closely related.
The present result isolates one formal instance of it,
gives the observed side a computable form,
and classifies the gap against an inventory of identity bases.

\subsection{Distinct Contribution}

The contribution is the conjunction of five elements:
\begin{enumerate}[label=(\arabic*)]
  \item the operational identity partition,
        the rule of sameness a deployed system applies,
        made computable over a disclosed mechanism registry;
  \item a surface-control condition
        under which a difference in identity treatment
        is attributable to a named mechanism,
        together with the record-identifier exclusion
        that keeps the test from firing on every correct system;
  \item faithfulness as a refinement relation
        between declared and operational partitions,
        with a finite witness refuting it,
        available for every regime in the inventory;
  \item a four-way lattice classification of divergence,
        when the sibling is informative inside the declared classes,
        separated from the global condition of regime substitution
        that it does not imply,
        with a finite version-field construction
        occupying the sub-sibling case;
  \item non-monotonicity of the passing verdict
        under extension of the transformation history alone,
        which indexes a passing audit to a history rather than to a store.
\end{enumerate}

The audit is disclosure-relative at three distinct levels:
the artifacts represented in \(D_B\),
the surfaces included in \(A_B\),
and the identity-relevant uses identified in each \(U_d\).
Each level carries a completeness claim
that the procedure cannot discharge
and a finite witness that refutes it.

\section{Limits}
\label{oi.sec.limits}

The result is bounded in nine ways.

First, the audit is relative to the registry.
It evaluates disclosed or discovered mechanisms
and does not prove that no further identity-relevant mechanism
exists outside the examined boundary or will arise in future behavior.
Definition~\ref{oi.def.omitted} states the object
whose later discovery refutes a completeness claim,
and Remark~\ref{oi.rem.registry} states why
the registry constraint cannot be dropped.

Second, the audit is relative to the surfaces it evaluated.
Registry completeness does not entail that every surface
over the registered artifacts was placed in \(A_B\).
A use jointly controlled by two registered fields may factor through neither alone,
so the surface carrying it may split a declared co-reference class
while each constituent field, taken alone, splits nothing.
The audit may not have included that joint surface in \(A_B\),
so every evaluated singleton surface can pass while the omitted joint surface fails.
Definition~\ref{oi.def.familycomplete} states the claim
and Definition~\ref{oi.def.omittedsurface} the object that refutes it.

Third, the audit is relative to the identified uses of each evaluated surface.
A surface cleared by the procedure
may carry an identity-relevant use the audit did not see,
and Proposition~\ref{oi.prop.refutes-use}
shows that such a use overturns the verdict.
The operational partition is computed from disclosed behavior,
and it is the rule of sameness the mechanism applies
only under the claim of Definition~\ref{oi.def.usecomplete}.

Fourth, the declared partition depends on
the accuracy and completeness of the family-labeled history,
which Definition~\ref{oi.def.precond} states as a precondition.
An inaccurate or incomplete history prevents the artifact from representing the examined transformations,
and the correct verdict is \INDET{}.

Fifth, the classification is relative to the imported inventory,
and it is local to the declared co-reference classes.
When the sibling also splits an examined declared class,
the four cases of Definition~\ref{oi.def.classification}
locate a divergence against that sibling
on pairs the declaration merges,
and say nothing about pairs it separates.
Otherwise the divergence is unpositioned.
Sibling alignment does not establish regime substitution,
by Proposition~\ref{oi.prop.aligned-not-substitution},
and no case of the classification names the basis
the implementation carries.
Identifying that basis remains work for the auditor.

Sixth, faithfulness is one-directional by design.
A surface that merges records the declaration separates
is not reported as a fault,
for the reason given in Remark~\ref{oi.rem.asymmetry}.
Conflation of declared-distinct referents at the level of the store
is a real failure and is outside the scope of a single-surface audit.

Seventh, verdicts are indexed to \((R,H)\).
Proposition~\ref{oi.prop.nonmonotone} shows
that a passing verdict can be overturned
by extending the history alone,
with the record domain held fixed.
A surface faithful on the examined artifact
may diverge on a larger one,
and the classification of a divergence may change with it.

Eighth, the procedure decides no downstream substantive question.
It does not choose the correct identity basis,
and either sibling may be appropriate for a given system.
It does not determine causal truth,
normative correctness,
legal interpretation,
legal authority,
responsibility,
or institutional legitimacy.
It preserves a prior condition for examining those questions:
stable and inspectable reference to what the claims concern.

Ninth, the worked application is illustrative rather than empirical.
The paper does not evaluate the audit against a deployed record system.
Such an evaluation would test the practical adequacy of the disclosed registry,
examined surfaces, identified uses, and record-indexed encodings,
rather than the finite decidability result established here.

\section{Conclusion}
\label{oi.sec.conclusion}

A record system answers the co-reference question twice:
once in what it declares
and once in what it does.
The declared answer partitions the records into co-reference classes.
The implementation supplies a second partition through its identity-relevant treatment of those records.
This paper formalizes that second answer as the operational identity partition
and compares it with the declared partition.

The comparison is a refinement test.
A mechanism is faithful when it splits no declared co-reference class,
and a single divergence witness refutes faithfulness.
Faithfulness composes across the examined surface family:
a failure on any surface refutes store faithfulness,
while the store passes only when every examined surface is faithful.

When faithfulness fails and the sibling basis also splits an examined declared class,
the divergence is classified against that alternative basis
in one of four sibling-relative cases.
Otherwise the divergence is unpositioned.
The classification is local to the declared co-reference classes.
Alignment with the sibling inside those classes
does not establish that the system implements the sibling regime,
and a single witness does not establish alignment.
A version counter incremented on every textual edit
splits declared classes strictly more finely than either imported basis.
The worked version-field example therefore inhabits the sub-sibling case.

The audit is decidable, three-valued, and disclosure-relative at three distinct levels:
the artifacts represented in \(D_B\), the surfaces included in \(A_B\),
and the identity-relevant uses identified in each \(U_d\).
Each level carries a completeness claim that the procedure cannot discharge and a finite witness that refutes it.

A passing verdict is also history-relative.
Because declared co-reference is a transitive closure,
extending the history alone,
with the record domain and operational treatment unchanged,
can merge declared classes
and create a witness among records already examined.
A passing audit is a statement about a history,
not a certificate about a store.

\emph{Neutral Substrates} states the constraint on shared foundational commitments.
\emph{Referential Regimes} states how identity persists under transformation.
Both describe what a system should declare.
The present result supplies the other half of the comparison:
a computable account of the identity relation induced by the disclosed implementation,
together with explicit limits on how far that account reaches.

\section*{Statements and Declarations}

\subsection*{Acknowledgments}
The author thanks the anonymous reviewers of earlier papers in this series
for comments that improved the framing, organization, and presentation of this work.
This research made use of SciX,
a scientific literature search and discovery platform,
whose related-literature tools helped identify relevant work across disciplinary boundaries.

\subsection*{Author Contributions}
The author is the sole contributor to this work
and is responsible for all aspects of the research, authorship, and publication.

\subsection*{Code Availability}
A reference implementation that checks the finite constructions of this paper
is available~\citep{case2026oiverification}.
It provides two independent checkers of the finite comparison procedure,
an explicit quadratic reference and the near-linear labelled procedure
of Theorem~\ref{oi.thm.decidable},
and confirms they agree on the worked examples (Example~\ref{oi.ex.loc-obj-fork},
Propositions~\ref{oi.prop.aligned-not-substitution}
and~\ref{oi.prop.witness-not-alignment},
Remark~\ref{oi.rem.nonvacuity}, the Section~\ref{oi.sec.worked} example,
and Proposition~\ref{oi.prop.nonmonotone}) and across small and randomized instances.
The implementation does not discover
mechanisms, evaluate a deployed system, or discharge the completeness
claims of Section~\ref{oi.sec.disclosure}.

\subsection*{Use of AI-Assisted Tools}
AI-assisted tools were used for editing, formatting, and consistency checking.
The author reviewed all suggestions and is solely responsible for the content.

\subsection*{Declaration of Conflicting Interest}
The author declares no potential conflicts of interest.

\end{document}